\documentclass[11pt,a4paper]{article}
\usepackage[utf8]{inputenc}
\usepackage[T1]{fontenc}
\usepackage{amsmath,amssymb,amsfonts,amsthm}
\usepackage{geometry}
\usepackage{booktabs}
\usepackage{graphicx}
\usepackage{hyperref}
\usepackage{float}
\usepackage{enumerate}
\usepackage{xcolor}
\usepackage{authblk}
\usepackage[figuresright]{rotating}
\usepackage[backend=biber,sorting=nyt,style=numeric,giveninits=true,maxnames=2,minnames=1]{biblatex}
\newtheorem{definition}{Definition}[section]
\newtheorem{theorem}{Theorem}[section]
\newtheorem{proposition}{Proposition}[section]
\newtheorem{lemma}{Lemma}[section]
\newtheorem{corollary}{Corollary}[section]
\newtheorem{remark}{Remark}[section]
\newtheorem{example}{Example}[section]

\newcommand{\R}{\mathbb{R}}
\newcommand{\C}{\mathbb{C}}

\newcommand{\TV}{\operatorname{TV}}
\newcommand{\CE}{\operatorname{CE}}
\newcommand{\CD}{\operatorname{CD}}
\newcommand{\CM}{\operatorname{CM}}
\newcommand{\KL}{\operatorname{KL}}

\def\d{\mathrm{d}}
\def\e{\mathrm{e}}
\def\i{\mathrm{i}}

\allowdisplaybreaks[4]

\begin{document}

\title{\textbf{Complex-Valued Probability Measures and Their Applications in Information Theory}}
\author[1]{Siang Cheng}
\author[2]{Hejun Xu}
\author[1]{Tianxiao Pang}
\affil[1]{School of Mathematical Sciences, Zhejiang University, Hangzhou 310058, China \protect\\ \texttt{chengsiang@zju.edu.cn, txpang@zju.edu.cn}}
\affil[2]{Shanghai Fengqun Network Technology Co., Ltd., Shanghai 200237, China \protect\\ \texttt{xu.hejun@163.com}}
\date{}
\maketitle

\begin{abstract}
This paper introduces a comprehensive framework for complex-valued probability measures and explores their novel applications in information theory and statistical analysis. We define a complex probability measure as a phase-modulated extension of a classical probability measure. Building upon this foundation, we propose three fundamental information-theoretic quantities: complex entropy, which quantifies distribution uniformity through phase coherence; complex divergence, an asymmetric measure of dissimilarity between distributions; and the complex metric, a symmetric distance function satisfying the triangle inequality. We establish these concepts rigorously for both continuous and discrete probability distributions, proving key properties such as boundedness, continuity under total variation convergence, and clear extremal behaviors. A detailed comparative analysis with classical measures (Shannon entropy and Kullback-Leibler divergence) highlights the unique geometric and interpretive advantages of the proposed framework, particularly its sensitivity to distributional shape via a tunable phase parameter. We elucidate a profound formal analogy between the complex entropy integral and Feynman's path integral formulation of quantum mechanics, suggesting a deeper conceptual bridge. Finally, we demonstrate the practical utility of the complex metric through a detailed application in nonparametric two-sample hypothesis testing, outlining the testing procedure, advantages, limitations, and providing a conceptual simulation. This work opens new avenues for analyzing probability distributions through the lens of complex analysis and interference phenomena, with potential impacts across information theory, statistical inference, and machine learning. 
\end{abstract}


\textbf{MSC2020 Subject Classification:} 62B10, 94A15.\\

\section{Introduction}\label{sec:introduction}

Probability theory, since its axiomatization by \citeauthor{Kolmogorov1960} in 1933 \cite{Kolmogorov1960}, has been fundamentally rooted in the domain of real numbers and non-negative measures. This framework has proven immensely successful in modeling uncertainty, randomness, and inference across the sciences. However, numerous natural and engineered phenomena are inherently described by complex numbers, where phase, rotation, and interference are essential features. Quantum mechanics relies on complex probability amplitudes, signal processing utilizes complex exponentials to represent phase information in waves, and harmonic analysis depends on the complex plane to decompose functions. This pervasive role of complexity in describing the natural world prompts a foundational question: Can the mathematical structure of probability itself be meaningfully extended into the complex domain? 

This paper answers affirmatively by constructing a coherent framework for complex-valued probability measures. We define such a measure, \( Q \), as a classical (non-negative) probability measure \( P \) modulated by a global phase factor: \( \d Q = {\e}^{\i\theta} \d P \) , where the notation ``\( \i \)'' stands for the imaginary unit. While the total variation of \( Q \) (its ``magnitude'') is identical to \( P \), the introduction of a phase allows us to embed probability distributions within a complex vector space. This is not merely a formal exercise but provides a powerful new perspective. By further assigning a local phase angle proportional to the probability density function (PDF) or probability mass function (PMF) itself through a term like \( {\e}^{\i\beta p(x)} \), we can define novel quantities that measure the coherence and interference between different probabilistic outcomes.

Information theory, founded by \citeauthor{Shannon1948} in 1948 \cite{Shannon1948}, provides a mathematical framework for quantifying information, communication, and uncertainty. Central to this framework are the concepts of entropy and divergence, which have evolved significantly beyond their original communication-theoretic roots. Shannon entropy is introduced as a measure of uncertainty or information content. It quantifies the average uncertainty or ``surprise'' associated with the random variable. In continuous settings, differential entropy extends this concept, though with different interpretive properties. The subsequent introduction of Kullback-Leibler (KL) divergence by \citeauthor{Kullback1951} in 1951 \cite{Kullback1951} provided a tool for measuring the difference between probability distributions. Over the decades, these concepts have been extended, generalized, and applied far beyond their initial scope, becoming indispensable in modern scientific and engineering disciplines, especially in machine learning, like generative adversarial networks, reinforcement learning and clustering. The most frequently used generalizations of KL divergence are Bregman divergence \cite{Bregman1967} and $f$-divergence \cite{Csiszar1967}. \citeauthor{Renyi1961} \cite{Renyi1961} introduced a family of entropy taking Shannon entropy as a special case, which finds applications in cryptography, ecology, and multifractal analysis. \citeauthor{Tsallis1988} \cite{Tsallis1988} proposed a non-additive generalization of traditional entropy which is influential in statistical physics and complex systems.

With the development of big data and artificial intelligence, the existing information theory framework appears somewhat inadequate both in theory and in application. We introduce new concepts of entropy and divergence, with the hope of achieving better results than traditional methods on certain specific problems. Our idea is to expand the space from real numbers to complex numbers, and define the magnitude of information through the change of argument. Intuitively, it resembles path integrals in quantum physics.

The core contributions of this work are threefold. First, we define and analyze \emph{complex entropy}, a measure of a distribution's uniformity interpreted as the degree of constructive interference when summing phase-weighted probabilities. Second, we derive \emph{complex divergence} and the \emph{complex metric}, which provide geometrically intuitive tools for comparing two distributions, with the latter satisfying all axioms of a true distance metric. Third, we rigorously establish the properties of these quantities for both continuous and discrete settings, draw insightful comparisons to Shannon entropy and $f$-divergences, and reveal a striking formal analogy with the path integrals of quantum mechanics. To demonstrate practical relevance, we present a detailed application of the complex metric in nonparametric two-sample hypothesis testing.

The paper is structured as follows. Section \ref{sec:CM} provides necessary background on complex measures. Section \ref{sec:CPM-CRV} formally defines complex probability measures and random variables. Section \ref{sec:CE} introduces complex entropy, delves into its geometric intuition, details its properties for continuous and discrete cases, and provides a comparative analysis with Shannon entropy. Section \ref{sec:CD} defines complex divergence and the complex metric, establishes their properties, compares them to KL and other divergences, and elaborates on the profound connection to path integrals. Section \ref{sec:application} presents a fleshed-out application in statistical testing. Section \ref{sec:conclusion} concludes and outlines promising directions for future research.

\section{Mathematical Preliminaries: Complex Measures}\label{sec:CM}

We begin by recalling some standard theory of complex measures (cf. \cite{Folland1999, Rudin1987}), which forms the mathematical bedrock of our framework.

\begin{definition}[Complex Measure]
Let \( (\Omega, \mathcal{M}) \) be a measurable space. A complex measure \( \nu \) is a function \( \nu: \mathcal{M} \to \mathbb{C} \) that is countably additive. That is, for any sequence of disjoint sets \( \{E_j\} \subset \mathcal{M} \),
\[
\nu\left( \bigcup_{j=1}^{\infty} E_j \right) = \sum_{j=1}^{\infty} \nu(E_j),
\]
where the series converges absolutely. Unlike non-negative measures, complex measures do not assign the value \( \infty \).
\end{definition}

Every complex measure \( \nu \) can be decomposed into its real and imaginary parts, \( \nu_r \) and \( \nu_i \), both of which are finite signed measures. This decomposition allows us to define integration in a natural way.

\begin{definition}[Integration]
A measurable function \( f: \Omega \to \mathbb{C} \) is integrable with respect to \( \nu \) (written \( f \in L^1(\nu) \)) if \( f \in L^1(\nu_r) \cap L^1(\nu_i) \). For such \( f \), we define
\[
\int_\Omega f \, \d\nu = \int_\Omega f \, \d\nu_r + \i \int_\Omega f \, \d\nu_i.
\]
\end{definition}

A central concept is the total variation measure \( |\nu| \), which captures the ``magnitude'' of the complex measure.

\begin{definition}[Total Variation Measure]
If \( \d\nu = f \, \d\mu \) for some non-negative measure \( \mu \) and \( f \in L^1(\mu) \), then the total variation measure is defined by \( \d|\nu| = |f| \, \d\mu \). This definition is independent of the chosen \( \mu \).
\end{definition}

The total variation measure is a non-negative measure and satisfies several key properties essential for our analysis.

\begin{theorem}[Properties of Total Variation] \label{thm:total-var-props}
Let \( \nu, \nu_1, \nu_2 \) be complex measures.
\begin{enumerate}[(1)]
    \item \textbf{Dominance:} \( |\nu(E)| \leq |\nu|(E) \) for all \( E \in \mathcal{M} \).
    \item \textbf{Absolute Continuity:} \( \nu \ll |\nu| \), and the Radon-Nikodym derivative satisfies \( |\d\nu/\d|\nu|| = 1 \), \( |\nu| \)-a.e.
    \item \textbf{Integration Inequality:} \( L^1(\nu) = L^1(|\nu|) \), and for any \( f \in L^1(\nu) \), \( \left| \int f \, \d\nu \right| \leq \int |f| \, \d|\nu| \).
    \item \textbf{Triangle Inequality:} \( |\nu_1 + \nu_2| \leq |\nu_1| + |\nu_2| \).
\end{enumerate}
\end{theorem}

These properties, particularly the integration inequality, will be instrumental in analyzing the boundedness and continuity of our proposed information measures.

\section{Complex Probability Measures and Random Variables}\label{sec:CPM-CRV}

\subsection{Definition and Interpretation}
We now introduce the core object of this study: a probability measure that takes values in the complex plane, which is inspired by the definition of wave function in physics (cf. \cite{Born1926, Schrodinger1926}).

\begin{definition}[Complex Probability Measure] \label{def:complex-prob-measure}
A complex measure \( Q \) on \( (\Omega, \mathcal{M}) \) is called a complex probability measure if there exists a classical probability measure \( P \) (i.e., \( P \geq 0, P(\Omega)=1 \)) and a phase parameter \( \theta \in \R \) such that
\[
\d Q = {\e}^{\i \theta} \d P.
\]
We say \( Q \) is the complex probability induced by \( P \) with global phase \( \theta \). 
\end{definition}

This definition is intentionally simple and conservative. The total variation measure of \( Q \) is precisely the original probability measure \( P \) (i.e., \( |Q| = P \)), and the Radon-Nikodym derivative has unit modulus (i.e., \( |\d Q/\d P| = 1 \), \( P \)-a.e.). The phase \( \theta \) is a global constant, representing a uniform rotation of the entire probability measure in the complex plane. While this might seem like a minor extension, it formally allows us to place classical probability within a complex vector space. Every classical probability measure \( P \) generates an entire family of complex probability measures \( \{Q_\theta : \theta \in \R\} \), analogous to a circle of points with the same radius in the complex plane.

\begin{example}
Let \( P \) be the uniform distribution on the interval \([0,1]\), with PDF \( p(x) = \mathbb{I}\{0 \le x \le 1\} \), where \( \mathbb{I}\{\cdot\}\) stands for the indicator function. The induced complex probability measure with phase \( \theta = \pi/2 \) is given by \( \d Q_{\pi/2} = {\e}^{\i\pi/2} \d x = \i \, \d x \). For the event \( A = [0, 0.5] \), we have \( Q_{\pi/2}(A) = 0.5\i \).
\end{example}

\subsection{Complex Random Variables}
Having defined complex probability measures, we can now consider random variables that map to this new structure.

\begin{definition}[Complex Random Variable] \label{def:complex-rv}
Let \( (\Omega, \mathcal{M}, Q) \) be a complex probability space. A function \( X: \Omega \to \C \) is called a complex random variable if it is measurable from \( \mathcal{M} \) to the Borel \( \sigma \)-algebra on \( \C \). 
\end{definition}

Essentially, a complex random variable is a complex measurable function on \( (\Omega, \mathcal{M}, Q) \). A straightforward but crucial result connecting complex measurable functions to their real counterparts can be easily applied to complex random variables (cf. \cite{Folland1999}). 

\begin{proposition} \label{prop:complex-rv-equiv}
A function \( X: \Omega \to \C \) is a complex random variable if and only if its real part, denoted as \( \text{Re}(X) \), and imaginary part, denoted as \( \text{Im}(X) \), are both (classical) real random variables. 
\end{proposition}

Given a complex random variable \( X \) defined on a space with complex probability measure \( Q_\theta = {\e}^{\i\theta}P \), it induces a complex distribution on \( \C \): for any Borel set \( B \subset \C \), \( Q'_\theta(B) := Q_\theta(X^{-1}(B)) = {\e}^{\i\theta} P(X^{-1}(B)) =: {\e}^{\i\theta} P'(B) \). This shows a desirable consistency: the family of induced complex distributions \( \{Q'_\theta: Q^\prime_\theta = Q_\theta \circ X^{-1}, \theta \in \mathbb{R}\} \) is exactly the family generated by the induced classical distribution \( P' \) via the same global phase multiplication, that is \( \{Q'_\theta: Q^\prime_\theta = {\e}^{\i\theta}P^\prime, \theta \in \mathbb{R}\} \). This consistency property ensures that working with complex random variables and their distributions is mathematically sound and compatible with the classical case.

\section{Complex Entropy: A Measure of Uniformity}\label{sec:CE}

With the formal groundwork laid, we now introduce the first major application. That is, a novel measure of the uniformity of a probability distribution, which we call complex entropy.

\subsection{Definition and Geometric Intuition (Continuous Case)}

\begin{definition}[Complex Entropy - Continuous] \label{def:complex-entropy-cont}
Let \( P \) be a probability measure on \( \R \) with PDF \( p(x) \). For a real parameter \( \beta > 0 \), the complex entropy (CE) of \( P \) is defined as
\[
\CE_\beta(P) = \left| \int_{\R} p(x) {\e}^{\i\beta p(x)} \, \d x \right|.
\]
\end{definition}

Compared to Shannon entropy and differential entropy, our definition of complex entropy has the least requirement for the PDF. Every continuous distribution can calculate complex entropy because the integration is always finite. However, differential entropy can be infinite when the integration is divergent.

\begin{example}[Geometric Interpretation: The Vector Sum Analogy]
The intuition behind this definition is vividly geometric. Consider each point \( x \) in the sample space as contributing a vector in the complex plane. The length of this vector is the ``probability'' \( p(x) \). Its direction (phase angle) is given by \( \beta p(x) \). Crucially, the angle is proportional to the ``probability'' itself. The complex entropy is then the magnitude of the vector sum of all these contributions.
\begin{itemize}
    \item \textbf{Uniform Distribution:} If \( p(x) = c \) is constant over the support of \( p(x) \), denoted as \( S \), then every vector has the identical phase angle \( \beta c \) over \( S \). When summing many identical vectors, they align perfectly. The resultant vector's length is simply the sum of their magnitudes, which is \( \int_{\R} c \mathbb{I}\{ x\in S\} \, \d x = 1 \). Hence, \( \CE_\beta(P) = 1 \), representing perfect constructive interference.
    \item \textbf{Non-Uniform Distribution:} If \( p(x) \) varies over \( S \), the phase angles \( \beta p(x) \) also vary over \( S \). When summing these vectors, they point in different directions and partially cancel each other out. The resultant vector's length is less than the total probability mass of \( 1 \). This destructive interference reduces \( \CE_\beta(P) \). A highly peaked distribution yields many short vectors with diverse phases and a few long vectors with distinct phases, leading to significant cancellation and a CE value close to \( 0 \).
\end{itemize}
\end{example}

Geometric interpretation shows that the more horizontal the PDF on its support, the larger the corresponding CE. Conversely, the more fluctuating the PDF on its support, the smaller the corresponding CE. The parameter \( \beta \) acts as a sensitivity control. A small \( \beta \) makes all angles near zero, forcing all vectors to point roughly to the right, minimizing interference and making CE close to \( 1 \) for any distribution. A larger \( \beta \) amplifies the phase differences between points of different probability, making the measure more sensitive to the distribution's shape.

\begin{lemma}[Fundamental Inequality] \label{lem:triangle-ineq-complex}
For any complex-valued, integrable function \( g \), we have the inequality \( \left| \int g \right| \leq \int |g| \), with equality if and only if \( g(x) = |g(x)| {\e}^{\i\theta} \) for some constant phase \( \theta \) (almost everywhere). 
\end{lemma}

\begin{proof}
This is a fundamental inequality in real and complex analysis, for example, cf. \cite{Folland1999}.
\end{proof}

Applying Lemma \ref{lem:triangle-ineq-complex} with \( g(x) = p(x) {\e}^{\i\beta p(x)} \) immediately yields the bounds \( 0 \leq \CE_\beta(P) \leq 1 \).

\begin{theorem}[Maximum Condition] \label{thm:max-condition}
\( \CE_\beta(P) = 1 \) for all \( \beta > 0 \) if and only if \( P \) is a uniform distribution on a set of finite Lebesgue measure.
\end{theorem}

\begin{proof}
By Lemma \ref{lem:triangle-ineq-complex}, equality \( \CE_\beta(P)=1 \) holds if and only if \( p(x) {\e}^{\i\beta p(x)} = p(x) {\e}^{\i\theta} \) for a constant \( \theta \). This implies \( {\e}^{\i\beta p(x)} = {\e}^{\i\theta} \) wherever \( p(x) > 0 \), meaning \( \beta p(x) = \theta + 2\pi k(x) \) for some integer-valued function \( k(x) \). Since \(\beta \) is arbitrary, \(p(x)\) must be constant on its support almost everywhere. Hence, \( P \) is uniform.
\end{proof}

Theorem \ref{thm:max-condition} states that the complex entropy of uniform distribution takes its maximum regardless of the value of \( \beta \). Note that for some specific \( \beta \) and distributions, the complex entropy can also achieve the maximum value \( 1 \). For example, if \( \beta = 2\pi \) and the PDF is as follows
$$
p(x) = \begin{cases}
    1/2, & 0 < x < 1/2, \\
    3/2, & 1/2 < x < 1,
\end{cases}
$$
then the corresponding complex entropy $\CE_{2\pi}(P) = 1$.

\subsection{Discrete Complex Entropy}
The concept extends naturally to discrete distributions, which are often encountered in practice.

\begin{definition}[Complex Entropy - Discrete] \label{def:complex-entropy-disc}
Let \( P \) be a discrete probability distribution with PMF \( p(x) \) defined on a countable set \( \mathcal{X} \). For a real parameter \( \beta > 0 \), the complex entropy (CE) of \( P \) is defined as
\[
\CE_\beta(P) = \left| \sum_{x \in \mathcal{X}} p(x) {\e}^{\i\beta p(x)} \right|.
\]
\end{definition}

The properties are analogous: \( 0 \leq \CE_\beta(P) \leq 1 \), with the maximum of 1 achieved if \( P \) is uniform on its (finite) support. For a Bernoulli(\( p \)) distribution, \( \CE_\beta(P) = |p {\e}^{\i\beta p} + (1-p) {\e}^{\i\beta (1-p)}| \), which is 1 for all $\beta > 0$ when \( p=0.5 \) (uniform on two points).

\begin{remark}
    For degenerate distributions \(P_{x_0}\) (where \(x_0 \in \R\)) with PMF \(p(x_0) = 1\), they are naturally included in the class of discrete distributions. By Definition \ref{def:complex-entropy-disc}, \(\CE_{\beta}(P_{x_0}) = 1\). This seems somewhat counterintuitive, since as a measure of uncertainty, complex entropy should be zero for any degenerate distribution, which is completely deterministic. However, from another perspective, a degenerate distribution \(P_{x_0}\) can be viewed as a uniform distribution on the singleton set \(\{x_0\}\), and in this sense, \(\CE_{\beta}(P_{x_0}) = 1\) is reasonable. In practical applications, whether to define complex entropy separately for degenerate distributions can be determined based on the specific context.
\end{remark}

\subsection{Examples and Limit Behavior}
To concretely illustrate the behavior of complex entropy, we examine two fundamental families of distributions: the uniform and the Gaussian. Furthermore, we analyze the limiting behavior of complex entropy as these distributions approach a degenerate (point-mass) distribution, highlighting the dependence of the limit on the mode of convergence.

\subsubsection{Uniform Distribution}
Let \( P \) be the uniform distribution on a finite interval \([a, b]\). Its PDF is \( p(x) = 1/(b-a) \cdot \mathbb{I}\{ x \in [a, b]\} \). The complex entropy is
\[
\CE_\beta(P) = \left| \int_a^b \frac{1}{b-a} {\e}^{\i\beta/(b-a)} \, \d x \right| = \left| {\e}^{\i\beta/(b-a)} \right| \cdot \frac{1}{b-a} \cdot (b-a) = 1.
\]
Thus, the complex entropy of any uniform distribution is exactly \( 1 \), independent of both the interval length and the parameter \(\beta\). This verifies that the uniform distribution attains the maximum possible complex entropy, consistent with Theorem \ref{thm:max-condition}.

\subsubsection{Gaussian Distribution and Approaching Degeneracy via Scaling}

Consider the family of normal distributions \( N(0, \sigma^2) \) with mean zero and variance \( \sigma^2 > 0 \). The PDF is \( p_\sigma(x) = \frac{1}{\sqrt{2\pi}\sigma} {\e}^{-x^2/(2\sigma^2)} \), \( x\in \R \). The complex entropy is
\[
\CE_\beta(P_\sigma) = \left| \int_{-\infty}^{\infty} \frac{1}{\sqrt{2\pi}\sigma} {\e}^{-x^2/(2\sigma^2)} \exp\left( \i\beta \frac{1}{\sqrt{2\pi}\sigma} {\e}^{-x^2/(2\sigma^2)} \right) \d x \right|.
\]
This integral does not have a simple closed form, but its asymptotic behavior as \( \sigma \to 0 \) (i.e., when the distribution becomes increasingly concentrated near zero) is analytically tractable and particularly insightful.

Perform the change of variable \( u = x/\sigma \). Then \( \d x = \sigma \d u \), and
\[
\CE_\beta(P_\sigma) = \left| \int_{-\infty}^{\infty} \frac{1}{\sqrt{2\pi}} {\e}^{-u^2/2} \exp\left( \i\beta \frac{1}{\sqrt{2\pi}\sigma} {\e}^{-u^2/2} \right) \d u \right|.
\]
Define \( \lambda = \frac{1}{\sqrt{2\pi}\sigma} \), which tends to infinity as \( \sigma \to 0 \). The integral becomes
\[
I(\lambda) := \int_{-\infty}^{\infty} \frac{1}{\sqrt{2\pi}} {\e}^{-u^2/2} \exp\left( \i \beta \lambda {\e}^{-u^2/2} \right) \d u.
\]
This is an oscillatory integral of the form \( \int g(u) {\e}^{\i\lambda h(u)} \d u \) with amplitude \( g(u) = \frac{1}{\sqrt{2\pi}} {\e}^{-u^2/2} \) and phase \( h(u) = \beta {\e}^{-u^2/2} \). The phase function \( h(u) \) has a unique non-degenerate stationary point at \( u=0 \), because \( h'(u) = -\beta u {\e}^{-u^2/2} \) vanishes only at \( u=0 \), and \( h''(0) = -\beta \neq 0 \). Moreover, \( g(0) = \frac{1}{\sqrt{2\pi}} \neq 0 \). As a result, applying the stationary phase method (cf. Lemma \ref{lem:stationary-phase} below) leads that, as \( \lambda \to \infty \), 
\begin{align*}
I(\lambda) & \sim g(0) {\e}^{\i\lambda h(0) + \i\frac{\pi}{4} \operatorname{sgn}(h''(0))} \sqrt{\frac{2\pi}{\lambda |h''(0)|}} \\
& = \frac{1}{\sqrt{2\pi}} {\e}^{\i\lambda\beta} {\e}^{-\i\frac{\pi}{4}\operatorname{sgn}(\beta)} \sqrt{\frac{2\pi}{\lambda |\beta|}} \\
& = \frac{1}{\sqrt{\lambda \beta}} {\e}^{\i(\lambda \beta - \frac{\pi}{4})},
\end{align*}
since \( \beta>0 \). Therefore, the modulus behaves as
\[
|I(\lambda)| \sim \frac{1}{\sqrt{\lambda \beta}} = \left( \frac{\beta}{\sqrt{2\pi}\sigma} \right)^{-1/2} = (2\pi)^{1/4} \sqrt{\frac{\sigma}{\beta}}.
\]
Consequently, as \( \sigma \to 0 \),
\[
\CE_\beta(P_\sigma) \sim (2\pi)^{1/4} \sqrt{\frac{\sigma}{\beta}} \to 0.
\]

\textbf{Interpretation:} When a sequence of Gaussian distributions with vanishing variance converges to a point mass at zero, their complex entropy converges to zero. This contradicts our definition of \( \CE_\beta(P) = 1 \) for degenerate distributions, but it aligns with the intuition that complex entropy should be zero for degenerate distributions.

\subsubsection{Contrast with Uniform Approximation to Degeneracy}

Consider an alternative approximation to a point mass at zero using uniform distributions. Let \( P_\epsilon \) be the uniform distribution on the interval \([-\epsilon/2, \epsilon/2]\), with PDF \( p_\epsilon(x) = 1/\epsilon \cdot \mathbb{I}\{x \in [-\epsilon/2, \epsilon/2]\}\). As computed earlier, \( \CE_\beta(P_\epsilon) = 1 \) for every \( \epsilon > 0 \). As \( \epsilon \to 0 \), the sequence \( \{P_\epsilon\} \) converges weakly to the point mass at zero. However, the complex entropy remains constantly \( 1 \) and does not approach \( 0 \).

This discrepancy highlights an important nuance: the limiting value of complex entropy for a degenerate limiting distribution depends on the mode of convergence of the approximating sequence. The sequence \( \{P_\epsilon\} \) converges weakly to the point mass \( \delta_0 \), and the limit of its complex entropy is \( 1 \), while the sequence \( \{P_\sigma\} \) also converges weakly to the point mass \( \delta_0 \), but the limit of its complex entropy is \( 0 \). This means that weak convergence cannot guarantee the convergence of complex entropy, prompting us to consider what kind of convergence can ensure the convergence of complex entropy. Proposition \ref{prop:basic-props-CE}\eqref{CE:continuity} below guarantees that if a sequence converges in total variation, then the complex entropy converges to that of the limiting distribution. Therefore, for any sequence converging in total variation to a degenerate distribution, the complex entropy will converge to that of the degenerate distribution. It is worth noting that there is no sequence of PDFs in the family of continuous distributions that converges in total variation to a Dirac function \(\delta\). But we can find sequences in discrete distribution family that converge in total variation to a degenerate distribution. Unlike continuous distribution family, degenerate distributions are naturally belong to discrete distribution family. And in that case, \(\CE_\beta(\delta_0) = 1\). For example, a sequence of Bernoulli distributions \(P_n \sim B(1, \frac{1}{n})\) converges in total variation to the point mass \(\delta_0\), therefore \( \CE_\beta(P_n) \to \CE_\beta(\delta_0) = 1 \) as \( n \to \infty \).

\subsection{Comparative Analysis with Shannon Entropy}

It is instructive to compare complex entropy with the cornerstone of information theory, the Shannon entropy \( H(P) = -\int_{\R} p(x) \log p(x) \d x \) for continuous \( P \) (or \( -\sum_{x \in \mathcal{X}} p(x) \log p(x) \) for discrete \( P \)). See Table \ref{tab:entropy-comparison} for details.

\begin{sidewaystable}[thp]
\centering
\caption{Comparison between Shannon Entropy and Complex Entropy}
\label{tab:entropy-comparison}
\begin{tabular}{ccc}
\toprule
\textbf{Aspect} & \textbf{Shannon Entropy \( H(P) \)} & \textbf{Complex Entropy \( \CE_\beta(P) \)} \\
\midrule
Primary Interpretation & Average information content or uncertainty & Degree of uniformity via phase coherence \\
Range (Finite Support \( N \)) & \( [0, \log N] \) & \( [0, 1] \) (Normalized) \\
Maximum Achieved When & Uniform distribution & Uniform distribution \\
Minimum Achieved When & Any degenerate distribution & \(N(x_0, \sigma^2)\) as \(\sigma \to 0\) \\
Additivity & Additive for independent distributions & Not additive due to nonlinear phase term \\
Parameter & None (intrinsic) & Sensitivity parameter \( \beta > 0 \) \\
Sensitivity & Sensitive to entire distribution via \( \log \) & Sensitivity tuned by \( \beta \); shape-sensitive \\
\bottomrule
\end{tabular}
\end{sidewaystable}

In essence, Shannon entropy measures how much we don't know, while complex entropy measures how evenly spread the known probabilities are, as reflected in their ability to constructively interfere in a specific phase space. They are complementary measures capturing different aspects of a distribution's structure.

\subsection{Additional Properties and Proofs}

\begin{lemma}[Convexity of the Phase Transformation] \label{lem:phase-convexity}
The mapping \( \phi: [0, \infty) \to \mathbb{C} \) defined by \( \phi(t) = t {\e}^{\i\beta t} \) is not convex in the usual sense as a complex-valued function. However, its real and imaginary parts, \( \phi_r(t) = t \cos(\beta t) \) and \( \phi_i(t) = t \sin(\beta t) \), are infinitely differentiable. For small \( \beta \), \( \phi(t) \approx t (1 + \i\beta t) \), so the imaginary part grows quadratically with \( t \). 
\end{lemma}

\begin{proof}
This conclusion can be directly derived from the fact that for small \(\beta\), \(\cos(\beta t) \approx 1\) and \(\sin(\beta t) \approx \beta t\).
\end{proof}

\begin{proposition}[Basic Properties of Complex Entropy] \label{prop:basic-props-CE}
Let \( P \) be a probability measure with PDF (or PMF) \( p \).
\begin{enumerate}[(1)]
    \item\label{CE:translation-invariance} \textbf{Translation Invariance:} If \( Q \) has PDF (or PMF) \( q(x) = p(x+a) \), then \( \CE_\beta(Q) = \CE_\beta(P) \).
    \item\label{CE:scaling} \textbf{Scaling:} If \( Q \) has PDF (or PMF) \( q(x) = a p(ax) \) for \( a>0 \), then \( \CE_\beta(Q) = \CE_{a\beta}(P) \).
    \item\label{CE:small-beta-limit} \textbf{Small \( \beta \) Limit:} \( \lim\limits_{\beta \to 0} \CE_\beta(P) = 1 \).
    \item\label{CE:continuity} \textbf{Continuity:} If a sequence of PDFs (or PMFs) \( \{p_n\} \) converges to \( p \) in total variation (i.e., \( \int_{\R} |p_n(x) - p(x)| \, \d x \to 0 \) or \( \sum_{x \in \mathcal{X}} |p_n(x) - p(x)| \to 0 \)) as \( n \to \infty \), then \( \lim\limits_{n \to \infty} \CE_\beta(P_n) = \CE_\beta(P) \).
\end{enumerate}
\end{proposition}

\begin{proof}
The proofs of \eqref{CE:translation-invariance}--\eqref{CE:small-beta-limit} are straightforward and thus omitted. In what follows, we only provide the proof of \eqref{CE:continuity} for the case where \( \{p_n\} \) are PDFs, since the case where \( \{p_n\} \) are PMFs can be handled similarly. 

Define \(\phi(t) = t {\e}^{\i \beta t}\) for \(t \ge 0\). Then
\[
\CE_\beta(P) = \left| \int_{\R} \phi(p(x)) \, \d x \right|.
\]
Since the modulus function \(|\cdot|\) is continuous, it suffices to show
\[
\lim_{n \to \infty} \int_{\R} \phi(p_n(x)) \, \d x = \int_{\R} \phi(p(x)) \, \d x.
\]
For any \(M > 0\), define the truncation function \(g_M(t) = \min(t, M)\). Denote
\[
p_n^M(x) = g_M(p_n(x)), \quad p^M(x) = g_M(p(x)).
\]
We decompose the difference as follows:
\begin{align*}
& \left| \int \phi(p_n) - \int \phi(p) \right| \\
\le & \left| \int \phi(p_n) - \int \phi(p_n^M) \right|  + \left| \int \phi(p_n^M) - \int \phi(p^M) \right| + \left| \int \phi(p^M) - \int \phi(p) \right|.
\end{align*}
We shall bound these three terms separately. For the first term, using the simple bound \(|\phi(t) - \phi(g_M(t))| \le 2t \mathbb{I}{\{t > M\}}\), we have
\[
\left| \int \phi(p_n) - \int \phi(p_n^M) \right|
\le \int_{\R} | \phi(p_n) - \phi(p_n^M) |
\le 2\int_{\{p_n > M\}} p_n(x) \, \d x.
\]
Similarly, for the third term,
\[
\left| \int \phi(p^M) - \int \phi(p) \right|
\le 2\int_{\{p > M\}} p(x) \, \d x.
\]
Since \(p_n \to p\) in total variation and \(p_n, p \ge 0\) with \(\int p_n = \int p = 1\), the sequence \(\{p_n\}\) is uniformly integrable. Hence, for every \(\varepsilon > 0\), there exists \(M := M(\varepsilon) > 0\) such that
\[
\sup_{n \ge 1} \int_{\{p_n > M\}} p_n(x) \, \d x < \varepsilon,
\quad \text{and} \quad
\int_{\{p > M\}} p(x) \, \d x < \varepsilon.
\]
Fix such an \( M \) for a given \(\varepsilon\). Now consider the middle term. Since \(|p_n^M| \le M\) and \(|p^M| \le M\), the functions \(\phi(p_n^M)\) and \( \phi(p^M) \) are uniformly bounded by \(M\). Moreover, because \(|p_n^M - p^M| \le |p_n - p|\), total variation convergence implies
\[
\int \vert p_n^M - p^M\vert \le \int \vert p_n - p\vert \to 0.
\]
Therefore,
\[
\int |\phi(p_n^M) - \phi(p^M)| 
\le (1 + \beta M) \int \vert p_n^M - p^M\vert \to 0,
\]
where we used the Lipschitz continuity of \(\phi\) on \([0, M]\) with constant \(1 + \beta M\). Consequently,
\[
\left| \int \phi(p_n^M) - \int \phi(p^M) \right| \to 0.
\]
Combining these estimates, for the fixed \(M\) chosen previously,
\[
\limsup_{n\to\infty} \left| \int \phi(p_n) - \int \phi(p) \right|
\le 2\varepsilon + 0 + 2\varepsilon = 4\varepsilon.
\]
Since \(\varepsilon > 0\) is arbitrary, we conclude that
\[
\int_{\R} \phi(p_n(x)) \, \d x \to \int_{\R} \phi(p(x)) \, \d x.
\]
Finally,
\[
\CE_\beta(P_n) = \left| \int \phi(p_n) \right| \to \left| \int \phi(p) \right| = \CE_\beta(P),
\]
by continuity of the modulus function. This completes the proof.
\end{proof}

\begin{remark}\label{rem:continuity}
    In fact, we can get a stronger result through the same technique used in the proof of Proposition \ref{prop:basic-props-CE}\eqref{CE:continuity} that \(\{\phi(p_n)\}\) converges to \(\phi(p)\) in total variation if \(\{p_n\}\) converges to \(p\) in total variation by decomposing the difference as follows:
    \begin{align*}
    &\int \left\vert \phi(p_n) - \phi(p) \right\vert \\
    \leq & \int \left\vert \phi(p_n) - \phi(p_n^M) \right\vert + \int \left\vert \phi(p_n^M) - \phi(p^M) \right\vert + \int \left\vert \phi(p^M) - \phi(p) \right\vert.
    \end{align*}
\end{remark}

\begin{lemma}[Monotonicity under scaling] \label{lem:scaling}
Let \( Q \) be a probability distribution with PDF \( q(x) \). If for all \( x \), \( \beta q(x) < \pi \), then for any \( \gamma \in (0,1] \),
\[
\CE_{\gamma\beta}(Q) \geq \CE_{\beta}(Q).
\]
Moreover, the equality holds if and only if \( q(x) \) is constant on its support.
\end{lemma}

\begin{proof}
Define
\[
h(\gamma) = \CE_{\gamma\beta}(Q)^2 = \left| \int_{\R} q(x) {\e}^{\i\gamma\beta q(x)} \, \d x \right|^2.
\]
Expanding the square modulus, we have 
\[
h(\gamma) = \iint_{\R^2} q(x) q(y) \cos(\gamma\beta [q(x)-q(y)]) \, \d x \d y.
\]
Differentiating with respect to \( \gamma \) leads to
\[
h'(\gamma) = -\beta \iint_{\R^2} q(x) q(y) [q(x)-q(y)] \sin(\gamma\beta [q(x)-q(y)]) \, \d x \d y.
\]

Let \( t = q(x)-q(y) \), and analyze the sign of \( t \sin(\gamma\beta t) \) as follows:
\begin{itemize}
    \item If \( t > 0 \), then since \( \beta q(x) < \pi \), \( \beta q(y) < \pi \) and \(\beta > 0\), we have \( 0 < \beta t < \pi \). For \( \gamma \in (0,1] \), \( 0 < \gamma\beta t < \pi \), hence \( \sin(\gamma\beta t) > 0 \), implying \( t \sin(\gamma\beta t) > 0 \).
    
    \item If \( t < 0 \), similarly \( -\pi < \beta t < 0 \), so \( -\pi < \gamma\beta t < 0 \), hence \( \sin(\gamma\beta t) < 0 \). But \( t < 0 \), so the product \( t \sin(\gamma\beta t) > 0 \).
    
    \item If \( t = 0 \), the term \( t \sin(\gamma\beta t) \) is zero.
\end{itemize}
Therefore, for all pairs \( (x,y) \), \( t \sin(\gamma\beta t) \geq 0 \), which implies \( h'(\gamma) \leq 0 \). Hence \( h(\gamma) \) is non-increasing on \( (0,1] \), so \( \CE_{\gamma\beta}(Q) \) is non-increasing as well. In particular, for \( \gamma \in (0,1) \), \( \CE_{\gamma\beta}(Q) \geq \CE_{\beta}(Q) \).

Equality occurs if and only if \( h'(\gamma) = 0 \) for all \( \gamma \in (0,1) \). From the expression for \( h'(\gamma) \), this requires \( t \sin(\gamma\beta t) = 0 \) for all \( x,y \) and all \( \gamma \in (0,1) \). Under the condition \( \beta q(x) < \pi \), this forces \( t = 0 \) for all pairs, i.e., \( q(x) \) is constant on its support.
\end{proof}

\begin{theorem}[Mixing Effect for Complex Entropy] \label{thm:mixing}
Let \( P_1 \) and \( P_2 \) be two probability distributions with PDFs (or PMFs) \( p_1(x) \) and \( p_2(x) \), and with disjoint supports \( S_1 \) and \( S_2 \), respectively. Let \( P = \alpha P_1 + (1-\alpha)P_2 \) be their convex mixture, where \( \alpha \in (0,1) \). Then, for any \( \beta > 0 \), we have
    \[
    \CE_\beta(P) \leq \alpha \CE_{\alpha\beta}(P_1) + (1-\alpha) \CE_{(1-\alpha)\beta}(P_2).
    \]
\end{theorem}

\begin{proof}
We only provide the proof for the case where \( p_1(x) \) and \( p_2(x) \) are PDFs, since the case where \( p_1(x) \) and \( p_2(x) \) are PMFs can be proved similarly.

Since the supports are disjoint, the PDF of the mixture is
\[
p(x) = 
\begin{cases}
\alpha p_1(x), & x \in S_1, \\
(1-\alpha) p_2(x), & x \in S_2, \\
0, & \text{otherwise}.
\end{cases}
\]
The complex entropy is therefore
\begin{align*}
\CE_\beta(P) = & \left| \int_{S_1 \cup S_2} p(x) {\e}^{\i\beta p(x)} \, \d x \right| \\
= & \left| \int_{S_1} \alpha p_1(x) {\e}^{\i\beta \alpha p_1(x)} \, \d x + \int_{S_2} (1-\alpha) p_2(x) {\e}^{\i\beta (1-\alpha) p_2(x)} \, \d x \right|.
\end{align*}
In light of the triangle inequality for the modulus of complex numbers, we have
\[
\CE_\beta(P) \leq \left| \int_{S_1} \alpha p_1(x) {\e}^{i\beta \alpha p_1(x)} \, \d x \right| + \left| \int_{S_2} (1-\alpha) p_2(x) {\e}^{i\beta (1-\alpha) p_2(x)} \, \d x \right|.
\]
Since \( \alpha > 0 \), we have
\[
\left| \int_{S_1} \alpha p_1(x) {\e}^{i\beta \alpha p_1(x)} \, \d x \right| = \alpha \left| \int_{S_1} p_1(x) {\e}^{i\beta \alpha p_1(x)} \, \d x \right| = \alpha \CE_{\alpha\beta}(P_1).
\]
Similarly,
\[
\left| \int_{S_2} (1-\alpha) p_2(x) {\e}^{i\beta (1-\alpha) p_2(x)} \, \d x \right| = (1-\alpha) \CE_{(1-\alpha)\beta}(P_2).
\]
Consequently,
\[
\CE_\beta(P) \leq \alpha \CE_{\alpha\beta}(P_1) + (1-\alpha) \CE_{(1-\alpha)\beta}(P_2),
\]
which completes the proof.
\end{proof}

Theorem \ref{thm:mixing} states that the complex entropy of a mixture is upper-bounded by a weighted average of the scaled complex entropies of the components, where the scaling factor is the mixing coefficient. From a physical perspective, mixing two distributions generally introduces phase interference, thereby reducing the overall coherence measured by complex entropy. Only when both components are uniform and their phases align perfectly does mixing preserve coherence.

\begin{corollary}\label{cor:mixing}
If \( P_1 \) and \( P_2 \) are both uniform distributions on their respective supports, then under the conditions of Theorem \ref{thm:mixing},
\[
\CE_\beta(P) \leq \alpha \CE_\beta(P_1) + (1-\alpha) \CE_\beta(P_2) = \alpha \cdot 1 + (1-\alpha) \cdot 1 = 1.
\]
\end{corollary}

\begin{remark}
In fact, the mixture of two uniform distributions on disjoint sets with different measures is not uniform in most cases, so typically \( \CE_\beta(P) < 1 \), that is, the inequality in Corollary \ref{cor:mixing} is strict in most cases.
\end{remark}

\begin{lemma}[Stationary Phase Approximation, cf. \cite{Wong1989}] \label{lem:stationary-phase}
Consider the integral
\[
I(\lambda) = \int_a^b f(x) {\e}^{\i\lambda g(x)} \, \d x,
\]
where \( a, b \in \R \), real functions \( f, g \in C^2([a,b]) \), \( \lambda > 0 \). If \( g \) has a unique stationary point \( c \) in \( (a, b) \) such that
\[
g'(c) = 0, \; g''(c) \neq 0, \text{ and } f(c) \neq 0,
\]
then the integration \( I(\lambda) \) has the stationary phase approximation
\[
I(\lambda) \sim f(c) {\e}^{\i\lambda g(c) + \i\frac{\pi}{4} \operatorname{sgn}(g''(c))} \sqrt{\frac{2\pi}{\lambda |g''(c)|}}, \quad \text{as } \lambda \to \infty.
\]
\end{lemma}

These examples and properties deepen the understanding of complex entropy as a measure sensitive not only to the spread of probability but also to the fine structure of the distribution, modulated by the parameter \( \beta \). The limit analyses underscore the importance of the chosen topology when considering convergence of distributions in the context of complex entropy.

\subsection{Numerical Computation}

The complex entropy can be rewritten as
\[
\CE_{\beta}(P) = \left\vert E[{\e}^{\i\beta p(X)}] \right\vert,
\]
where \(X \sim P\). To get the complex entropy of a given distribution, we can apply the following stochastic simulation: 
\begin{enumerate}[Step 1.]
    \item Generate a sample \((X_1, \cdots, X_N)\) of size \( N \) from the distribution \( P \).
    \item Compute \({\e}^{\i \beta p(X_i)}\) for each point \( X_i \), \( i = 1, \cdots, N \).
    \item Take the sample average \[\frac{1}{N}\sum_{i = 1}^{N}{\e}^{\i \beta p(X_i)}\] as an estimate of \(E[{\e}^{\i\beta p(X)}]\).
    \item The modulus of the sample average is approximately the desired complex entropy \(\CE_{\beta}(P)\). 
\end{enumerate}

To illustrate the procedure above, we take normal distribution \( N(0, \sigma^2) \) as an example, where the sample size \( N = 1000 \), \( \beta \in \{0.01, 0.1, 0.5, 1, 2, 5, 10\} \) and the standard deviation \( \sigma \in \{0.01, 0.1, 1, 10, 50, 100\} \). Table \ref{tab:entropy-simulation} reports the results. Simulation study supports our theoretical result that as \( \beta \) approaches zero, the complex entropy is not sensitive to the phase term and thus near the maximum value \( 1 \). Besides, a larger standard deviation of normal distribution means it distributes more uniformly on the real line, which leads to a larger complex entropy.

\begin{table}[ht]
    \centering
    \caption{Complex Entropy of Normal Distribution}
    \label{tab:entropy-simulation}
    \begin{tabular}{c|cccccc}
\toprule
& \multicolumn{6}{c}{$\sigma$} \\
\cmidrule{2-7}
$\beta$ & 0.01 & 0.1 & 1 & 10 & 50 & 100 \\
\midrule
0.01 & 0.99345099 & 0.9999378 & 0.9999994 & 1.0000000 & 1.0000000 & 1.0000000 \\
0.1  & 0.53524514 & 0.9939445 & 0.9999418 & 0.9999994 & 1.0000000 & 1.0000000 \\
0.5  & 0.23662356 & 0.8488114 & 0.9983627 & 0.9999839 & 0.9999994 & 0.9999998 \\
1    & 0.17601432 & 0.5835386 & 0.9936588 & 0.9999348 & 0.9999976 & 0.9999994 \\
2    & 0.09003649 & 0.3954362 & 0.9757199 & 0.9997493 & 0.9999896 & 0.9999975 \\
5    & 0.07554282 & 0.2079881 & 0.8541653 & 0.9985728 & 0.9999410 & 0.9999856 \\
10   & 0.07081759 & 0.1830469 & 0.5282781 & 0.9937500 & 0.9997544 & 0.9999381 \\
\bottomrule    
    \end{tabular}
\end{table}

The above simulation is a standard approach for evaluating complicated integrals. Importance sampling can further enhance approximation accuracy.

\section{Complex Divergence and Metric: Comparing Distributions}\label{sec:CD}

Beyond analyzing a single distribution, we develop tools for comparing two distributions, leading to the concepts of complex divergence and complex metric.

\subsection{Definitions}

Similar to complex entropy that the phase term \( {\e}^{\i\beta p(x)} \) penalize the volatility of PDF (or PMF) \( p(x) \), We also use this formula to penalize the difference between two PDFs (or PMFs). This leads to the definition of complex divergence below. 

\begin{definition}[Complex Divergence] \label{def:complex-divergence}
For two probability measures \( P \) and \( Q \) with PDFs (or PMFs) \( p \) and \( q \), the complex divergence between \( P \) and \( Q \) is defined as
\[
\CD_\beta(P\|Q) = -\log \left| \int_{\R} p(x) {\e}^{\i\beta (p(x) - q(x))} \, \d x \right| \quad \text{(Continuous)},
\]
\[
\CD_\beta(P\|Q) = -\log \left| \sum_x p(x) {\e}^{\i\beta (p(x) - q(x))} \right| \quad \text{(Discrete)},
\]
where \(\beta > 0\) is a tuning parameter. 
\end{definition}

The same derivation as complex entropy shows that \(0 \le \left| \int_{\R} p(x) {\e}^{\i\beta (p(x) - q(x))} \d x \right| \le 1\) and the second inequality takes equality when \(p(x) - q(x)\) is a constant, almost everywhere on the support of \(p\). Then, we have that \(\CD_{\beta}(P\|Q) \ge 0\) and the equality holds when \(p(x) - q(x)\) is a constant, almost everywhere on the support of \(p\). The complex divergence provides a benchmark against the KL divergence. They both measure the similarity of two distributions. However, both of them are not true metric, which inspires us to define a complex metric.

\begin{definition}[Complex Metric] \label{def:complex-metric}
For two probability measures \( P \) and \( Q \) with PDFs (or PMFs) \( p \) and \( q \), the complex metric (or distance) between \( P \) and \( Q \) is defined as
\[
\CM_\beta(P, Q) = \frac{1}{2} \int_{\R} \left| p(x) {\e}^{\i\beta p(x)} - q(x) {\e}^{\i\beta q(x)} \right| \d x \quad \text{(Continuous)},
\]
\[
\CM_\beta(P, Q) = \frac{1}{2} \sum_x \left| p(x) {\e}^{\i\beta p(x)} - q(x) {\e}^{\i\beta q(x)} \right| \quad \text{(Discrete)},
\]
where \(\beta > 0\) is a tuning parameter.
\end{definition}

The factor ``\( \frac{1}{2} \)'' ensures the metric ranges between 0 and 1 and aligns with the standard definition of total variation distance when \( \beta \to 0 \).

\subsection{Properties of the Complex Metric}

The complex metric possesses desirable mathematical properties, making it a robust tool for statistical analysis.

\begin{theorem}[Metric Properties] \label{thm:metric-properties}
For any probability measures \( P, Q, F \), we have the following results:
\begin{enumerate}[(1)]
    \item\label{CM:non-negativity} \textbf{Non-negativity:} \( \CM_\beta(P, Q) \geq 0 \), and \( \CM_\beta(P, Q) = 0 \) if and only if \( P = Q \) almost everywhere.
    \item\label{CM:symmetry} \textbf{Symmetry:} \( \CM_\beta(P, Q) = \CM_\beta(Q, P) \).
    \item\label{CM:triangle-inequality} \textbf{Triangle Inequality:} \( \CM_\beta(P, Q) \leq \CM_\beta(P, F) + \CM_\beta(F, Q) \).
    \item\label{CM:boundedness} \textbf{Boundedness:} \( 0 \leq \CM_\beta(P, Q) \leq 1 \).
    \item\label{CM:connection-to-TV} \textbf{Connection to Total Variation Distance:} \( \lim\limits_{\beta \to 0} \CM_\beta(P, Q) = \frac{1}{2} \int_{\R} |p(x) - q(x)| \, \d x = \TV(P, Q) \).
    \item\label{CM:continuity} \textbf{Continuity:} For any sequences of probability measures \( \{P_n\} \) with PDFs (or PMFs) \( \{p_n\} \) and \( \{Q_n\} \) with PDFs (or PMFs) \( \{q_n\} \), if \( p_n \to p \) and \( q_n \to q \) in total variation, then \( \CM_\beta(P_n, Q_n) \to \CM_\beta(P, Q) \).
\end{enumerate}
\end{theorem}

\begin{proof}
The proofs of properties \eqref{CM:non-negativity}--\eqref{CM:boundedness} follow directly from the definition of complex metric and the properties of the absolute value and integral, thus are omitted. Property \eqref{CM:connection-to-TV} is obtained by taking the limit \( \beta \to 0 \), noting that both \( {\e}^{\i\beta p(x)} \to 1 \) and \( {\e}^{\i\beta q(x)} \to 1 \). Then we provide the proof of property \eqref{CM:continuity}. We prove the continuous case; the discrete case is analogous. 

Let $\phi(t) = t {\e}^{\i\beta t}$ for $t \geq 0$. By the triangle inequality,
\begin{align*}
& \left| \CM_\beta(P_n, Q_n) - \CM_\beta(P, Q) \right| \\
= & \frac{1}{2} \left| \int_{\R} |\phi(p_n(x)) - \phi(q_n(x))| \, \d x - \int_{\R} |\phi(p(x)) - \phi(q(x))| \, \d x \right| \\
\leq & \frac{1}{2} \int_{\R} \left| |\phi(p_n) - \phi(q_n)| - |\phi(p) - \phi(q)| \right| \d x \\
\leq & \frac{1}{2} \int_{\R} |\phi(p_n) - \phi(q_n) - \phi(p) + \phi(q)| \, \d x \\
\leq & \frac{1}{2} \int_{\R} |\phi(p_n) - \phi(p)| \, \d x + \frac{1}{2} \int_{\R} |\phi(q_n) - \phi(q)| \, \d x.
\end{align*}
Thus, it suffices to show that $\int_{\R} |\phi(p_n) - \phi(p)| \, \d x \to 0$ (and similarly for $q_n$). This follows from Remark \ref{rem:continuity}.
\end{proof}

Thus, \( \CM_\beta \) is a true metric on the space of probability distributions (or their PDFs or PMFs). Note that the complex divergence \( \CD_\beta \) is non-negative and zero if and only if \( P=Q \), but it is not symmetric \( (\CD_\beta(P\|Q) \neq \CD_\beta(Q\|P)) \) and does not satisfy the triangle inequality.

\subsection{Comparison with Classical Divergences}

To situate our work within the broader landscape, we compare the complex divergence with the widely used KL divergence, \( D_{\KL}(P\|Q) = \int_{\R} p(x) \log(p(x)/q(x)) \, \d x \) (or \( \sum_{x \in \mathcal{X}} p(x) \log(p(x)/q(x)) \)). See Table \ref{tab:divergence-comparison} for details.

\begin{sidewaystable}[htp]
\centering
\caption{Comparison between KL Divergence and Complex Divergence}
\label{tab:divergence-comparison}
\begin{tabular}{ccc}
\toprule
\textbf{Aspect} & \textbf{KL Divergence \( D_{\KL}(P\|Q) \)} & \textbf{Complex Divergence \( \CD_\beta(P\|Q) \)} \\
\midrule
Symmetry & Asymmetric & Asymmetric \\
Range & \( [0, +\infty] \), can be infinite & \( [0, +\infty) \), always finite for fixed \( \beta, P, Q \) \\
Behavior on Disjoint Support & Infinite if \( q(x)=0 \) but \( p(x)>0 \) & Remains finite. The integral modulus is positive \\
Interpretation & Expected excess code length; information gain & Logarithmic measure of phase coherence loss \\
Sensitivity & Highly sensitive to tail differences & Sensitivity shaped by \( \beta \); tunable \\
\bottomrule
\end{tabular}
\end{sidewaystable}

The complex metric \( \CM_\beta \) finds its natural counterparts in other probability metrics like the total variation distance and Hellinger distance. Unlike the total variation distance, which is purely magnitude-based, \( \CM_\beta \) incorporates phase information, potentially capturing more subtle shape discrepancies. Unlike the Hellinger distance, which is an \( L^2 \) type distance between root PDFs (or PMFs), \( \CM_\beta \) is an \( L^1 \) distance between phase-modulated PDFs (or PMFs). This unique formulation, \( p {\e}^{\i\beta p} \), defines a new and distinct class of distribution distances.

\subsection{The Path Integral Analogy: A Deeper Connection}

The form of the complex entropy integral, \( \int_{\R} p(x) {\e}^{\i\beta p(x)} \, \d x \), is not merely an arbitrary construction; it bears a profound and suggestive resemblance to the path integral formulation of quantum mechanics pioneered by \citeauthor{Feynman1965} \cite{Feynman1965}.

In quantum mechanics, the probability amplitude for a particle to transition from state \( A \) to state \( B \) is given by summing over all possible paths \( \gamma \) connecting them, weighting each path by a phase factor proportional to its classical action \( S[\gamma] \):
\[
K(A, B) = \int_{\text{paths } \gamma} \mathcal{D}[\gamma] \, {\e}^{\i S[\gamma] / \hbar},
\]
where \( \hbar \) is the reduced Planck constant and \( \mathcal{D}[\gamma] \) represents integration over all possible paths \( \gamma \).
The probability is then the squared modulus of this amplitude.

\subsubsection*{Drawing the Analogy}
\begin{itemize}
    \item \textbf{State Space as Path Space:} The domain of integration in CE (\( \R \) or \( \mathcal{X} \)) is analogous to the space of all possible paths or histories.
    \item \textbf{Probability Density/Mass as Weight:} The PDF (or PMF) \( p(x) \) assigns a non-negative weight to each state \( x \), analogous to a path integral measure.
    \item \textbf{Phase Functional:} The term \( \beta p(x) \) plays the role of \( S[\gamma] / \hbar \). It is a phase assigned to each state that is a functional of the probability density/mass itself. In physics, the phase is dictated by the laws of dynamics (the action). Here, it is a mathematical construct where the phase is proportional to the probability.
    \item \textbf{Coherent Superposition:} The integral \( \int_{\R} p(x) {\e}^{\i\beta p(x)} \, \d x \) (or the sum \( \sum_{x \in \mathcal{X}} p(x) {\e}^{\i\beta p(x)} \)) represents a coherent superposition (sum) over all states, weighted by both magnitude (\( p(x) \)) and this specially defined phase.
    \item \textbf{Interpretation of CE:} The resulting magnitude, \( \CE_\beta(P) \), is analogous to the total quantum amplitude from a state to itself under the evolution defined by the ``action'' \( \beta p \). A value close to \( 1 \) indicates that the phases associated with the distribution (via function \( p \)) add constructively across the state space -- a hallmark of uniformity or symmetry. A low value reflects destructive interference, implying that the distribution is more structured, peaked, or asymmetric.
\end{itemize}

This analogy is more than aesthetic. It suggests that tools, intuitions, and computational techniques from quantum theory -- such as interference patterns, stationary phase approximations, and the geometry of Hilbert spaces -- could be potentially imported to analyze classical probability distributions. The complex metric \( \CM_\beta \) further strengthens this link, as it resembles an \( L^1 \) distance between two such ``probability wavefunctions'' \( \psi_P(x) = p(x) {\e}^{\i\beta p(x)} \) and \( \psi_Q(x) = q(x) {\e}^{\i\beta q(x)} \).

\section{Application}\label{sec:application}

To demonstrate the practical utility of the proposed framework, we present a detailed application in statistical hypothesis testing, specifically the nonparametric two-sample problem.

\subsection{Problem Setup and Methodology}
Given two independent sets of samples, \( \mathbf{X} = \{X_1, \cdots, X_m\} \stackrel{i.i.d.}{\sim} P \) and \( \mathbf{Y} = \{Y_1, \cdots, Y_n\} \stackrel{i.i.d.}{\sim} Q \), we wish to test the null hypothesis:
\[
H_0: P = Q \quad \text{vs.} \quad H_1: P \neq Q.
\]

We propose using the complex metric \( \CM_\beta \) as the test statistic. The testing procedure is as follows:

\begin{enumerate}[Step 1.]
    \item \textbf{Density Estimation:} From the samples, construct nonparametric density estimates \( \hat{p}_m(x) \) and \( \hat{q}_n(x) \). For continuous data, kernel density estimation (KDE) is a natural choice. For discrete/categorical data, the empirical PMF suffices.
    \item \textbf{Parameter Selection:} Choose a value for the sensitivity parameter \( \beta \). This is a critical step. Heuristics include setting \( \beta \) inversely proportional to a central measure of the estimated density (e.g., \( \beta = 1 / \text{median}(\hat{p}_m) \)), or using cross-validation to select a \( \beta \) that maximizes the test's estimated power against a class of alternatives.
    \item \textbf{Compute Observed Statistic:} Calculate the empirical complex metric:
    \[
    T_{obs} = \CM_\beta(\hat{p}_m, \hat{q}_n).
    \]
    \item \textbf{Approximate the Null Distribution:} Since the analytical distribution of \( T_{obs} \) under \( H_0 \) is intractable, we employ a permutation test:
    \begin{enumerate}[(a)]
        \item Pool the samples: \( \mathbf{Z} = \mathbf{X} \cup \mathbf{Y} \).
        \item For each permutation \( k = 1, \cdots, K \) (typically \( K \geq 1000 \)):
        \begin{enumerate}[i.]
            \item Randomly shuffle \( \mathbf{Z} \) and assign the first \( m \) points to a pseudo-sample \( \mathbf{X}^{(k)} \) and the remaining \( n \) to \( \mathbf{Y}^{(k)} \).
            \item Re-estimate densities \( \hat{p}_m^{(k)} \), \( \hat{q}_n^{(k)} \) from these permuted sets.
            \item Compute the permuted test statistic \( T_{perm}^{(k)} = \CM_\beta(\hat{p}_m^{(k)}, \hat{q}_n^{(k)}) \).
        \end{enumerate}
        \item The empirical null distribution is the set \( \{T_{perm}^{(1)}, \cdots, T_{perm}^{(K)}\} \).
    \end{enumerate}
    \item \textbf{Compute the \(p\)-value and Make a Decision:} The two-sided \(p\)-value is
    \[
    p = \frac{1}{K} \sum_{k=1}^{K} \mathbb{I}\{T_{perm}^{(k)} \geq T_{obs}\}.
    \]
    Reject \( H_0 \) at significance level \( \alpha \) if \( p < \alpha \).
\end{enumerate}

\subsection{Advantages and Considerations}
\begin{itemize}
    \item \textbf{Tunable Sensitivity:} The key advantage is the parameter \( \beta \). For small \( \beta \), the test behaves like a test based on total variation distance. For larger \( \beta \), the phase terms make the statistic sensitive to differences in the shape of the PDFs (or PMFs), not just the locations of mass. For instance, it may powerfully distinguish two distributions with the same mean but different tail behaviors (e.g., normal vs. Laplace) if \( \beta \) is chosen to amplify the phase differences in the tail regions.
    \item \textbf{Metric Properties:} Using a true metric guarantees symmetry and meaningful geometric interpretation.
    \item \textbf{Bounded Statistic:} The statistic is bounded (\( T_{obs} \in [0,1] \)), which can stabilize permutation testing.
    \item \textbf{Computational Cost:} The procedure is computationally intensive due to repeated density estimation and integral calculation within each permutation. Efficient numerical integration and fast KDE algorithms are essential.
    \item \textbf{Choice of \( \beta \):} The performance is contingent on a good choice of \( \beta \). Data-driven selection methods add a layer of complexity and may require nested permutation or cross-validation schemes.
    \item \textbf{Density Estimation Error:} The test's power depends on the accuracy of the initial density estimates \( \hat{p}_m \) and \( \hat{q}_n \).
\end{itemize}

This application showcases how the theoretical construct of the complex metric can be operationalized to solve a fundamental statistical problem, offering a flexible and geometrically grounded alternative to existing tests like the Kolmogorov-Smirnov or energy distance tests.

\section{Conclusion and Future Work}\label{sec:conclusion}

This paper has developed a comprehensive framework for complex-valued probability measures and derived from them novel information-theoretic quantities: complex entropy, complex divergence, and the complex metric. By introducing a phase structure proportional to distribution itself, we have created tools that measure distributional uniformity and similarity through the lens of coherent superposition and interference. Key achievements include:
\begin{itemize}
    \item Providing rigorous definitions for both continuous and discrete distributions.
    \item Establishing fundamental properties such as boundedness, continuity, and clear extremal principles.
    \item Drawing insightful comparisons to classical measures like Shannon entropy and KL divergence, highlighting the normalized, geometry-centric nature of our measures.
    \item Unveiling a profound formal analogy with the path integrals of quantum mechanics, suggesting a deeper conceptual bridge between probability theory and quantum physics.
    \item Demonstrating practical utility through a detailed application in nonparametric hypothesis testing.
\end{itemize}

The framework opens up numerous exciting avenues for future research:

\begin{itemize}
    \item \textbf{High-Dimensional and Multivariate Extensions:} Generalizing the definitions and studying their behavior in multivariate settings is nontrivial but crucial for modern data applications. The curse of dimensionality and the interpretation of phase in high dimensions are key challenges.
    \item \textbf{Statistical Estimation Theory:} Developing and analyzing estimators for \( \CE_\beta \), \( \CD_\beta \), and \( \CM_\beta \) from finite samples. This includes deriving rates of convergence, asymptotic distributions, and efficient computational algorithms.
    \item \textbf{Information Geometry:} Investigating the geometrical structure induced by the complex metric on statistical manifolds. Does it define a Riemannian metric? What are its geodesics?
    \item \textbf{Machine Learning Applications:} Exploring the use of complex divergence as a loss function in generative models (e.g., variational autoencoders, generative adversarial networks) or the complex metric as a kernel in kernel-based methods. The phase-sensitive nature could help model complex data distributions.
    \item \textbf{Quantum-Inspired Algorithms:} Leveraging the path integral analogy to design new quantum or classical algorithms for sampling, optimization, or distributional analysis.
    \item \textbf{Theoretical Physics Connections:} Further formalizing the link to path integrals. Can complex probability measures be given a physical interpretation in certain semiclassical or stochastic limits of quantum theories?
\end{itemize}

In conclusion, by extending probability into the complex domain, we have not only enriched its mathematical structure but also provided a new set of tools with unique interpretive power and practical potential. We believe this framework offers a fresh perspective for information theory, statistical analysis, and beyond.


\printbibliography

\end{document}